\pdfoutput=1
\documentclass[reqno]{amsart}

\usepackage[sc]{mathpazo}
\usepackage{eulervm}

\usepackage[T1]{fontenc}

\usepackage{graphicx,mathrsfs}
\usepackage{bbm}
\usepackage[utf8]{inputenc}
\usepackage{hyperref}
\usepackage{cleveref}
\usepackage{tikz}
\usepackage{amssymb}

\usepackage[english]{babel}
\usepackage{verbatim}
\usepackage{comment}
\usepackage{todonotes}

\usepackage{enumerate}

\usepackage{fancyhdr}
\makeatletter
\let\ps@plain\ps@fancy
\makeatother
\newtheorem{lemma}{Lemma}[section]
\newtheorem{theorem}[lemma]{Theorem}

\theoremstyle{definition}
\newtheorem{definition}[lemma]{Definition}
\newtheorem{question}[lemma]{Question}

\theoremstyle{remark}
\newtheorem{remark}[lemma]{Remark}

\newcommand{\p}{\prime}

\def\N{\mathbb{N}}
\def\Z{\mathbb{Z}}

\def\prof{\mathcal{P}}

\def\sym{\mathcal{S}}
\def\M{\mathcal{M}}
\def\T{\mathcal{T}}

\newcommand\citep{\cite}
\newcommand{\numberedsection}[1]{\section{#1}}

\begin{document}

\title{Impossibility results on stability\\ of phylogenetic consensus methods}

\author{Emanuele Delucchi}
\author{Linard Hoessly}
\author{Giovanni Paolini}


\begin{abstract} We answer two questions raised by Bryant, Francis and Steel in their work on consensus methods in phylogenetics.
Consensus methods apply to every practical instance where it is desired to aggregate a set of given phylogenetic trees (say, gene evolution trees) into a resulting, ``consensus'' tree (say, a species tree). Various stability criteria have been explored in this context, seeking to model desirable consistency properties of consensus methods as the experimental data are updated (e.g., more taxa, or more trees, are mapped). However, such stability conditions can be incompatible with some basic regularity properties that are widely accepted to be essential in any meaningful consensus method.
Here, we prove that such an incompatibility does arise in the case of extension stability on binary trees and in the case of associative stability.
Our methods combine general theoretical considerations with the use of computer programs tailored to the given stability requirements.

\end{abstract}

\maketitle

\subsection*{Context}
The problem of merging the information carried by a set of phylogenetic trees into a resultant (``consensus'') tree is standard and well-studied. For instance, this problem arises as one tries to combine many gene trees in order to reconstruct a common species phylogeny, or when aggregating a set of estimates resulting from the application of different clustering algorithms to the same genomic data set. More generally, consensus methods have wide applications in biology \citep{degnan} as well as in other sciences, for example, social choice theory \citep{Arrow}.

This variety of applications has motivated a general axiomatic study of consensus methods. In biology, the field was pioneered by McMorris and collaborators, see \citep{McMorris_book} for a survey. Here, one of the research threads  is the study of ``stability conditions'' for consensus functions, which encode the requirement that a consensus method should be consistent under ``restriction'' of all input trees to a subset of taxa, see for example \citep{Barth}. For example, when computed on the branching structures induced on a specific subset of the taxa, the consensus method should output the branching structure induced by the consensus tree computed from the full data. A main question is whether such stability conditions are compatible with ``Pareto-type'' properties, where one requires that if some partial feature is shared among all trees we want to aggregate, then this feature should be present in the consensus tree as well.

\subsection*{Motivation and aim}
Our paper is motivated by two questions asked in a recent work of Bryant, Francis and Steel \citep{Steel_1}, who followed up on %
\citep{Steel_2}. In their paper, they carry out a detailed feasibility analysis of stability conditions that express ``future-proofing'' of phylogenetic trees, that is, consistency of consensus methods with respect to increase of experimental evidence. (For example, an increase of the set of taxa or an increase of the size of the set of trees from which the consensus is to be drawn.) 
Let us explain intuitively the four properties of consensus methods on which, following \citep{Steel_1}, we will focus (for precise definitions see %
Section \ref{sec:BG}).
\begin{itemize}
	\item[] {\bf Regularity} properties ensure that the output does not depend on the naming of taxa nor on the order of the trees. Moreover, if all input trees are equal, then the consensus tree should also be equal to the input trees.
	\item[] {\bf Extension stability} requires that, if the input data is updated by including a new taxon in each tree, the branching structure among the ``original'' taxa is preserved in the updated consensus tree. 
	\item[] {\bf Associative stability} allows, among other things, to reduce the computation of the consensus tree to a series of consensus problems between pairs of trees.
\end{itemize}
As is usual, these properties are considered together with a {Pareto}-type property which, again following \citep{Steel_1}, we take to be {\bf Pareto on rooted triples}. This means that if all input trees display the same nontrivial branching order when restricted to a specific triple of taxa, then the consensus tree must display the same branching order when restricted to the same taxa.

Bryant, Francis, and Steel conclude by stating two main open questions about the existence of consensus methods \citep[Concluding comments]{Steel_1}.  The first question asks whether there exist regular consensus methods that are extension stable when the input data are restricted to binary trees.
The second question asks whether there exist regular consensus methods that are Pareto on rooted triples and associatively stable. %
We answer both questions in the negative. 

Our proof consists in a reduction to a problem of integer linear programming (ILP), which we then solve. ILP has been used in the past in order to compute consensus trees, see \cite{dong2010}; however, ours seems to be the first application of ILP to a (non)existence proof in phylogenetics.

\numberedsection{Background}\label{sec:BG}
\subsection{Phylogenetic trees}\label{sec:pt}
Our setup mostly follows \citep{steel_book16},  and in particular we restrict our attention to \emph{rooted phylogenetic trees}. We fix a set (say, of taxa) $X$ and write $RP(X)$ for the set of rooted phylogenetic trees on the leaf set $X$.
A \emph{cluster} of a tree is any set of leaves that consists of all descendants of a particular vertex of the given tree. The set of clusters of a tree forms a hierarchy
(We call {\em hierarchy} any family of subsets of a given set such that any two elements in the family intersect trivially, that is, their intersection is either empty or equal to one of the two sets).
For every hierarchy on a set $X$ that contains $X$ itself and all singleton sets, but does not contain the empty set, there is a unique phylogenetic tree whose clusters form the given hierarchy.
In particular, two trees have the same associated hierarchy if and only if they are equivalent.

We say a tree $T^\p\in RP(X)$ \textit{refines} a tree $T\in RP(X)$, and write $T\preceq T^\p$, if the hierarchy of $T$ is contained in that of $T^\p$ (this means that every cluster of $T$ is also a cluster of $T^\p$).
This defines a partially ordered set $(RP(X),\preceq)$ whose maximal elements are given by the binary trees and whose unique minimal element is the ``star'' tree, where every leaf is adjacent to the root (the hierarchy of the star tree consists only of the singletons and $X$ itself).
Given a tree $T\in RP(X)$ and a subset $Y\subset X$, the restriction of $T$ to $Y$ is the tree $T|_Y\in RP(Y)$ obtained by restricting $T$ to the leaves in $Y$ (see Figure \ref{fig:restriction}).
\begin{figure}
	\begin{tikzpicture}[x=.9em,y=1em]
	\node (L) at (-7,1) {$T$:};
	\node (1) at (-5,0) {};
	\node (2) at (-3,0) {};
	\node (3) at (-1,0) {};
	\node (4) at (1,0) {};
	\node (5) at (3,0) {};
	\node (6) at (5,0) {};
	\node (1a) at (-5,-.7) {$1$};
	\node (2a) at (-3,-.7) {$2$};
	\node (3a) at (-1,-.7) {$3$};
	\node (4a) at (1,-.7) {$4$};
	\node (5a) at (3,-.7) {$5$};
	\node (6a) at (5,-.7) {$6$};
	\node (A) at (0,1) {};
	\node (B) at (-1,2) {};
	\node (C) at (-2,3) {};
	\node (D) at (0,5) {};
	\node (E) at (4,1) {};
	\draw (1.center) -- (C.center) -- (4.center);
	\draw (2.center) -- (B.center);
	\draw (3.center) -- (A.center);
	\draw (5.center) -- (E.center) -- (6.center);
	\draw (C.center) -- (D.center) -- (E.center);
	\end{tikzpicture}\quad
	\begin{tikzpicture}[x=.9em,y=1em]
	\node (L) at (-5,1) {$T\vert_{\{2,3,6\}}$:};
	\node (2) at (-2,0) {};
	\node (3) at (0,0) {};
	\node (6) at (2,0) {};
	\node (2a) at (-2,-.7) {$2$};
	\node (3a) at (0,-.7) {$3$};
	\node (6a) at (2,-.7) {$6$};
	\node (A) at (-1,1) {};
	\node (B) at (0,2) {};
	\draw (2.center) -- (A.center) -- (3.center);
	\draw (A.center) -- (B.center) -- (6.center);
	\end{tikzpicture}\quad
	\begin{tikzpicture}[x=.9em,y=1em]
	\node (L) at (-6,1) {$T\vert_{\{1,2,3,5\}}$:};
	\node (1) at (-3,0) {};
	\node (2) at (-1,0) {};
	\node (3) at (1,0) {};
	\node (5) at (3,0) {};
	\node (1a) at (-3,-.7) {$1$};
	\node (2a) at (-1,-.7) {$2$};
	\node (3a) at (1,-.7) {$3$};
	\node (5a) at (3,-.7) {$5$};
	\node (A) at (0,1) {};
	\node (B) at (-1,2) {};
	\node (C) at (0,3) {};
	\draw (2.center) -- (A.center) -- (3.center);
	\draw (1.center) -- (B.center) -- (A.center);
	\draw (B.center) -- (C.center) -- (5.center);
	\end{tikzpicture}
	\caption{A tree $T$ on the set of taxa $X=\{1,2,3,4,5,6\}$ and its restrictions $T\vert_Y$ for $Y=\{2,3,6\}$ and $Y=\{1,2,3,5\}$.}
	\label{fig:restriction}
\end{figure}
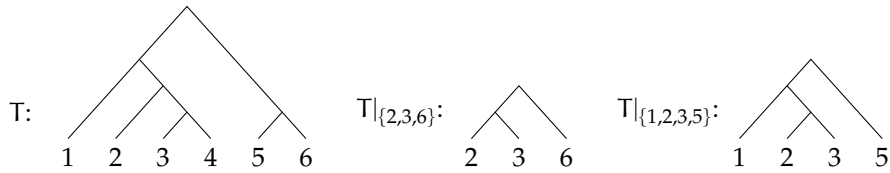
If $|Y|=3$ and $T|_Y$ is binary, we say that the \emph{rooted triple} $T|_Y$ is displayed by $T$. Notice that any rooted phylogenetic tree is fully determined by its set of rooted triples.

\subsection{Consensus functions and consensus methods}
\label{sec:consensus}

\begin{definition}\label{df:1234}
	Let $X$ be a finite set of taxa and $k \in \N$ a natural number.
	A \emph{profile} of trees is an element $(T_1,\cdots,T_k)\in RP(X)^k$.
	\begin{enumerate}
		\item A \textit{k-consensus function on $X$} is a function
		$$\varphi_X^k:RP(X)^k\to RP(X).$$
		\item A \textit{consensus function on $X$} is a function
		$$\varphi_X:\cup_{k\geq 1}RP(X)^k\to RP(X).$$
		\item A \textit{k-consensus method} is a function that, for every set $Y$ of taxa, associates with any profile $(T_1,\cdots,T_k)\in RP(Y)^k$
		a tree $\varphi_Y^k(T_1,\cdots,T_k)\in RP(Y)$.
		We consider such a method as a set of functions $\varphi_Y^k$, one for every $Y$, and denote it simply by $\varphi^k$.
		\item A \textit{consensus method} is a function that, for every set of taxa $Y$ and any $k\in \N$, associates with any profile $(T_1,\cdots,T_k)\in RP(Y)^k$
		a tree in $\varphi_Y(T_1,\cdots,T_k)\in RP(Y)$.
		We consider such a method as a set of $k$-consensus methods $\varphi^k$, one for every $k$, and denote it by $\varphi$.
	\end{enumerate}
\end{definition}

We will at times need to consider ($k$-)consensus methods where the only allowed sets of taxa are the subsets of a given finite set $X$: in this case, we will speak of a \emph{($k$-)consensus method on $X$}.

To summarize the terminology set-up: ``functions'' are rules that apply only to a fixed set of taxa, while ``methods'' do not have this restriction. The length of an input profile is not restricted, unless a prefix is added (as in "$k$-consensus").

\begin{remark}[Related definitions in the literature]
	Our use of the term ``consensus method'' conforms to  \citep{Bryant_classif,Steel_1}, whereas the term ``consensus function on $X$'' matches \citep{steel_book16}. However, the terminology is not completely consistent throughout the literature: in particular, we remark on some instances where the objects we introduced in Definition \ref{df:1234} appear under different names. In Day and McMorris' book \citep{McMorris_book}, our $k$-consensus functions on $X$ and consensus functions on $X$ are called consensus rules and complete consensus rules, respectively. In addition, the ``supertrees'' treated in  \citep{Steel_2} are analogues of consensus methods where one does not require the leaf sets of all trees in a profile to coincide.
\end{remark}

\subsection{Axiomatic requirements for consensus methods}\label{sec:reg}
We now recall some general axioms for consensus methods.
We follow \citep{Steel_1} and call a consensus method $\varphi$ \textbf{regular} if it satisfies the following three axioms:
\begin{enumerate}
	\item \textbf{Unanimity.} The value of $\varphi$ on any profile consisting of a $k$-fold repetition of a single tree $T$ is $T$ itself.

	\item \textbf{Anonymity.} Changing the order of the trees in a profile does not affect the value of $\varphi$ on it.
	\item \textbf{Neutrality.} Changing the labels on the leaves of the trees in a profile simply relabels the leaves of the consensus tree in the same way.
\end{enumerate}

Furthermore we say that a consensus method $\varphi$ on a set $X$ of taxa is \textbf{Pareto on rooted triples} if the following condition is satisfied for all $Y\subseteq X$ and all trees $T_1,\ldots,T_k\in RP(X)$, $T'\in RP(Y)$:
$$
\textrm{if } T'\preceq T_i\vert_Y \textrm{ for all }i=1,\ldots k, 
\textrm{ then } T'\preceq \varphi_X^k(T_1,\ldots,T_k)\vert_Y .
$$
An equivalent rephrasing in %
more colloquial terminology 
is that any rooted triple that is displayed by a set of trees must be displayed by their consensus tree as well.

\subsection{Examples}
\label{sec:examples}
In the following we mention a few consensus methods that have appeared in the literature, witnessing the existence of methods that do satisfy several combinations of the above-mentioned properties. In order to describe them precisely, we will make use of the characterization of trees by means of hierarchies of clusters, see above.
For more examples of consensus methods we refer to \citep{Bryant_classif,McMorris_book,steel_book16}. An overview is given in \cite[Figure 2]{Bryant_classif}.

A first class of consensus methods determines the hierarchy of clusters of the consensus tree based on the frequency of appearance of those clusters in the input trees. 
\textbf{Majority rule}, which is probably the most widely used consensus method in practice, returns the tree determined by the hierarchy of all clusters that appear in more than half of the input trees.
\textbf{Strict consensus} returns the tree given by the clusters that appear in every input tree, while \textbf{loose consensus} returns the tree defined by the set of all clusters that appear in at least one input tree and are compatible with the other input trees.

A second type of examples is of recursive nature. The idea is to associate to each profile of trees $\prof =(T_1,\cdots, T_k)\in RP(X)$ a partition $\Pi(\prof)$ of $X$ whose blocks will form the maximal clusters of the returned consensus tree. (A \emph{block} of a partition is simply one of its elements, that is, a subset of $X$.)
Then, for every block $B$ of $\Pi(\prof)$, one computes the partition $\Pi(\prof\vert_B)$ associated with the profile restricted to $B$, and so forth recursively.
The union of the blocks of all partitions is then the hierarchy of the consensus tree.
\textbf{Adams consensus} \citep{Adams} defines $\Pi(\prof)$ as the set of nonempty intersections of the maximal clusters of the trees in $\prof$. In \textbf{Aho consensus} \citep{Aho} (which is called local consensus in \citep{Bryant_classif}) the partition $\Pi(\prof)$ is the set of the connected components of the graph with vertex set $X$ and where a pair of vertices $\{a,b\}\subseteq X $ is joined by an edge if there is some $c\in X$ for which the rooted triple defined by the hierarchy $\{\{a,b\}, \{a\},\{b\},\{c\}\}$ is displayed by all trees in $\prof$.

All these consensus methods are regular. Aho consensus and Adams consensus are Pareto on rooted triples. As we will discuss later on, strict consensus is associatively stable (see Definition \ref{def:AS}), whereas majority consensus, loose consensus, Aho consensus, and Adams consensus are not \citep{Steel_1}.

\numberedsection{Results and methodology}
\label{sec:ReMe}

\subsection{Extension stability}\label{sec:ES}

This section focuses on consensus methods that satisfy the following condition, defined in \citep{Steel_2} and meant to encode the fact that a consensus should behave consistently with respect to passing to subsets of taxa.

\begin{definition}\label{df:ExSt}
	Fix a positive integer $k$, and let $\{\varphi^k_Y\}_{\emptyset\neq Y\subseteq X}$ be a $k$-consensus method on a set $X$ of taxa.
	This $k$-consensus method on $X$ is called {\bf extension stable} if, for all nonempty subsets $Y\subseteq X^\p\subseteq X$ and every profile $(T_1,\ldots,T_k)\in RP(X^\p)$,
	$$
	\varphi_Y^k(T_1\vert_Y,\ldots,T_k\vert_Y)\preceq
	\varphi_{X^\p}^k(T_1,\ldots,T_k)\vert_Y.
	$$
	A $k$-consensus method $\varphi^k$ is extension stable if, for all $X$, the $k$-consensus method $\{\varphi^k_Y\}_{\emptyset\neq Y\subseteq X}$ on $X$ is extension stable.
	A consensus method $\varphi$ is extension stable if for all $k\geq 1$ the $k$-consensus method $\varphi^k$ is extension stable.
\end{definition}

One of the main results of 
\citep{Steel_1} is that {\em no regular $2$-consensus method is extension stable}.
In the same paper, the feasibility of different relaxations of extension stability was discussed. The first question left open in \citep{Steel_1} is about extension stability under the restriction of the domain of consensus methods to binary trees. More precisely, fix a set $X$ of taxa and let $RBP(X)\subseteq RP(X)$ denote the subset of all rooted binary phylogenetic trees -- that is, the phylogenetic trees on $X$ where every internal vertex has exactly two children.

\begin{definition}
	A $k$-consensus method $\varphi^k_{\ast}$ on $X$ (resp.\ $k$-consensus method $\varphi^k$, consensus method $\varphi$) is {\em extension stable on binary trees} if the method obtained by restricting each  $\varphi^k_{\ast}$ (resp.\ $\varphi^k$, $\varphi$)  to the set of binary phylogenetic trees is extension stable (in the sense of Definition \ref{df:ExSt}, replacing $RP(X^\p)$ with $RBP(X^\p)$).
\end{definition}

\begin{question}[\cite{Steel_1}]\label{Q1} Is there a regular consensus method that is extension stable on binary trees? 
\end{question}

\begin{theorem}\label{thm:ES}
	There is no regular extension stable $k$-consensus method among profiles of binary trees on more than $4$ taxa, for any even profile size $k$. In particular, there is no regular and extension stable consensus method on profiles of binary trees.
\end{theorem}

The gist of the proof is a verification by means of a computer program that there is no extension stable $2$-consensus method on sets of $5$ taxa. We will give the details of the computation in the Appendix.
The sufficiency of this verification depends on the following, easily checked fact.

\begin{remark} 
	For every positive even integer $k$, any $k$-consensus function $\varphi^{k}_X$ on a set $X$ of taxa induces a $2$-consensus function $\varphi^2_X$ on $X$ by setting
	$$
	\varphi^2_X(T_1,T_2) := \varphi^{k}_X(
	\underbrace{T_1,\ldots, T_1}_{k/2 \textrm{ times}},
	\underbrace{T_2,\ldots, T_2}_{k/2 \textrm{ times}}
	).
	$$
	Regularity and extension stability of $\varphi^{k}_X$ are inherited by $\varphi^{2}_X$.
\end{remark}

\subsection{Associative stability}\label{sec:AS}

\begin{definition}[\cite{Steel_1}]
	\label{def:AS}
	Let $\varphi$ denote a consensus method on a set of taxa $X$. We say that $\varphi$ is {\bf associatively stable} if the following equality is satisfied for all $T_1,\ldots,T_k\in RP(X)$:
	$$
	\varphi_X^k (T_1,\ldots,T_k) = \varphi^2_X(\varphi_X^{k-1}(T_1,\ldots,T_{k-1}),T_k).
	$$
\end{definition}

In \citep{Steel_1} it is noted 
that Adams consensus is associatively stable when restricted to trees of height $2$, but not for trees of height $4$, and that Aho consensus is not associatively stable even for trees of height $2$. (The {\em height} of a rooted tree is the maximum distance between the root and any leaf.)  On the other hand, associative stability is satisfied by some elementary methods such as strict consensus, which, however, fails to be Pareto on rooted triples. This motivates the following.

\begin{question}[\cite{Steel_1}] Is there a regular consensus method that satisfies associative stability and is Pareto on rooted triples?
\end{question}

\begin{theorem}
	There exists no regular, associatively stable consensus method that is Pareto on rooted triples.
	\label{thm:AS}
\end{theorem}

As was already remarked in \citep{Steel_1}, if $\varphi$ is a regular and associatively stable consensus method on a set of taxa $X$, then $\varphi^2_X$ is a commutative, idempotent and associative binary operation on $RP(X)$. Thus it is enough to prove that such a binary operation does not exist.

\begin{lemma}\label{lem:C2} There exists no regular, associative $2$-consensus function which is Pareto on rooted triples for any set of $5$ or more taxa.
\end{lemma}

\begin{remark} The only regular and associatively stable consensus method on $3$ and $4$ taxa, which is Pareto on rooted triples, is Adams consensus. This is discussed in the Appendix.
\end{remark}

The proof of Lemma \ref{lem:C2} rests on a computational check of the case of $5$ taxa (see Appendix).
From there, the full generality follows via the following lemma.

\begin{lemma} Fix a positive integer $k$ and a set of taxa $X$. Every regular  $k$-consensus function on $X$ which is Pareto on rooted triples and associatively stable induces a $k$-consensus function on every subset of $X$ which is also regular, Pareto on rooted triples, and associatively stable.
\end{lemma}

\begin{proof} Fix a subset $Y\subseteq X$ and an enumeration $x_1,\ldots, x_l$ of the set $X\setminus Y$.
	Given any tree $T\in RP(Y)$, define a tree $T^X \in RP(X)$ as in Figure \ref{unnumbered}.
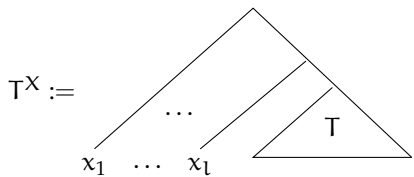
	\begin{figure}[h]
	\begin{center}
		\begin{tikzpicture}[x=1em,y=1em]
		\node (L) at (-8,3) {$T^{X}:=$};
		\node (R) at (0,6) {};
		\node (A) at (-6,0) {$x_1$};
		\node (Z) at (6,0) {};
		\node (T) at (3,3) {};
		\node (TT) at (3,1.5) {$T$};
		\node (Q) at (0,0) {};
		\node (PP) at (2,4) {};
		\node (B) at (-4,0) {$\ldots$};
		\node (B1) at (-2.8,2) {$\ldots$};
		\node (C) at (-2,0) {$x_l$};
		\draw (C.north) -- (PP.center);
		\draw (A.north) -- (R.center) -- (Z.north);
		\draw (T.center) -- (Q.north) -- (Z.north);
		\end{tikzpicture}
	\end{center}
	\caption{}\label{unnumbered}
	\end{figure}
	Notice that, for all $T_1,T_2\in RP(Y)$, 
	\begin{equation}\label{eq:inj}
	T_1^X = T_2^X \textrm{ if and only if } T_1=T_2.
	\end{equation}
	Now, given a consensus method $\varphi$ on $X$ we can define a consensus method $\psi_Y$ on $Y$ by setting, for every positive integer $k$,
	$$
	\psi^k_Y(T_1,\ldots,T_k):=\varphi_X^k(T_1^X,\ldots,T_k^X)\vert_Y.
	$$
	We immediately observe that regularity of $\varphi$ implies regularity of $\psi$. If $\varphi^k_X$ is Pareto on rooted triples, then 
	\begin{equation*}\label{eq:PRT}
	\varphi_X^k(T_1^X,\ldots,T_k^X)=(\psi^k_Y(T_1,\ldots,T_k))^X.
	\end{equation*}
	If in addition $\varphi$ is associatively stable, we can use this equation in order to write, for every $T_1,\ldots,T_k\in RP(Y)$,
	\begin{align*}
	\psi_Y^2(\psi_Y^{k-1}(T_1,\ldots,T_{k-1}),T_k)^X 
	& %
	= \varphi_X^2(\psi_Y^{k-1}(T_1,\ldots,T_{k-1})^X,T_k^X)\\
	& %
	= \varphi_X^2(\varphi_X^{k-1}(T_1^X,\ldots,T_{k-1}^X),T_k^X) \\
	& %
	=\varphi_X^k(T_1^X,\ldots,T_k^X)
	=\psi_Y^k(T_1,\ldots,T_k)^X.
	\end{align*}
	In view of Equation \eqref{eq:inj}, this proves associative stability of $\psi$.
\end{proof}

\numberedsection{Concluding discussion}

We have answered the two main questions left open in \citep{Steel_1}, about extension stability and associative stability of consensus methods on phylogenetic trees.
On the one hand, we have proved that, under widely accepted regularity  requirements, there cannot exist any consensus method that is stable under addition of taxa, even when the input trees are required to be binary (Theorem \ref{thm:ES}). We thus strengthen the result of Bryant, Francis and Steel. The meaning of this theorem is that, no matter which method is used in order to extract a consensus from a profile of binary trees, the branching structure in the consensus tree is not guaranteed to hold once the set of available taxa is enlarged -- even if the ``augmented'' input trees agree with the original profile when restricted to the previously available taxa.
Our other main result, Theorem \ref{thm:AS}, states that there is no associatively stable consensus method that satisfies some common regularity and Pareto-type properties. 
This means that, when enlarging the set of trees from which consensus is extracted, it may not be enough to compute the consensus between the new trees and the ``old'' consensus tree, and thus one is forced to carry out the computation anew, starting from the complete profile of trees. In fact, as pointed out in \citep{Steel_1}, there do exist consensus methods that satisfy associative stability: such methods however fail to simultaneously possess both basic regularity and Pareto properties. In this light, our result can be interpreted by saying that those basic properties are intrinsically complex -- and, in particular,  the substantial computational advantage that is granted by associative stability is ``too much to hope for''. 

We note that our considerations about rooted phylogenetic trees have implications for other types of data structures. First, our impossibility results apply also to consensus among {\bf unrooted trees}, see for example \cite[Section 3.17]{Husonetal} for the relevance of this case. Given a profile of rooted trees, one can consider a profile of unrooted trees obtained from the former by appending a "special leaf" to the root of all trees in the rooted profile. From any  regular (and stable) consensus tree among the unrooted trees we can then extract a regular (and stable) consensus tree for the rooted profile by deleting the special leaf. Thus, no regular and (extension- or associatively) stable consensus method can exist for unrooted trees on more than $6$ leaves. Moreover, {\bf supertree consensus} is a generalization of tree consensus, hence any claim of non-existence of tree-consensus functions implies in particular the same claim for supertree consensus.

\bigskip

As was also remarked in \citep{Steel_1}, such negative results are valuable inasmuch as they uncover the intrinsic limitations of certain approaches, thus helping direct future research towards feasible paths. We suggest four such possible directions of further research. %
\begin{itemize}
\item[1.] Prompted by discussions with researchers at the Swiss Institute of Bioinformatics, we propose to consider single splits as the basic features of phylogenetic trees, instead of rooted triples. A split in a (rooted) tree is a bipartition of the taxa that is obtained by deleting an edge of the tree.

\noindent {\bf Q.} Are there regular consensus methods on phylogenetic trees that are Pareto-optimal with respect to single splits and satisfy extension stability (with respect to single splits) or associative stability?

\item[2.] Another way forward may involve relaxing the constraint on the data structures outputted by consensus methods. The minimum hybridization network associated with a profile of trees is a rooted phylogenetic network that ``minimally deviates from being a tree'' among those  displaying all trees in the profile. We refer to \cite[Section 11.5]{Husonetal} for a precise description, and only mention here that the deviation from a true tree form is measured by the number $h$ of reticulations of the resulting network. Trees have no reticulation, that is, $h=0$.

Now notice that a hybridization network displays every triple that is displayed by any input tree. By contracting some edges on this network we can decrease reticulation, coming closer to a tree-network, at the cost of having to ``take a stance'' by resolving some rooted triples. Our results show that certain regularity, Pareto and stability conditions cannot be satisfied by a tree (that is, a non-reticulated network). It is natural to ask how reticulated consensus networks need to be, in order to behave in a stable fashion.

\noindent {\bf Q.} Find a bound on the required number of reticulations of a regular, extension stable consensus network (for example, as a function of the number of leaves of the trees of the input profile).

\item[3.]   Since performing inference often leads to a collection of different phylogenetic trees for the set of taxa under investigation, consensus methods are mostly applied to aggregate the given data in this context \citep{Husonetal}. Hence both the extension of consensus methods towards reasonable choices of random variables as well as their axiomatic limitations are an interesting avenue for future research. We have only considered deterministic consensus methods (cf.\ Definition \ref{df:1234}), and thus our impossibility results do not apply to probabilistic consensus methods such as greedy consensus, see for example \citep{degnan}. In particular we leave it as an interesting question to determine whether there exist probabilistic consensus methods satisfying (probabilistic) analogues of Pareto and stability properties.

\item[4.] Our methods combine theoretical reductions with explicit computations, and are in principle adaptable to explore other stability conditions that appeared in the literature, for instance in \citep{Barth}. Even if such conditions are sometimes dependent on each other (for example, our Theorem \ref{thm:ES} implies incompatibility of condition $(I6)$  in \citep{Barth} with regularity assumptions, even for binary trees), this approach is in its essence case-by-case.
From a systematic point of view, we believe  that it would be interesting to shed more light on general conceptual or structural obstructions to the existence of consensus methods with given properties. A possible instance of work in this direction is the following. Recall that an associatively stable regular consensus method is equivalent to a partial order on the set of inputs (for example, phylogenetic trees on $n$ taxa) that is a meet-semilattice \cite[Chapter 1, p.~7]{Gra} and equivariant with respect to the action on the permutation group $S_n$ on the labels of taxa. The existence of such a consensus method that is Pareto on rooted triples is then equivalent to the existence of a ($S_n$-equivariant) meet-semilattice structure $\vee_{?}$ on the set of phylogenetic trees such that, for all trees $T_1$, $T_2$,
\begin{equation}\label{eq:semil}
F(T_1 \vee_{?} T_2) \geq F(T_1) \vee F(T_2)
\end{equation}
where $F: RP(X)\to \prod_{A\subseteq [n], \vert A \vert =3} RP(A)$ is the canonical function that associates with a tree $T$ the set of rooted triples $(T\vert_A)_{A\subseteq [n], \vert A \vert =3}$, and $\leq$ is the natural partial order (resp.\ $\vee$ the meet operation) in the codomain of $F$. In fact, if $\vert A \vert = 3$, $RP(A)$ has a natural semilattice structure (the only one on four elements that is not a lattice and not a chain) and this defines naturally a ``product'' semilattice structure on $\prod_{A\subseteq [n], \vert A \vert =3} RT(A)$, see for example \cite[Chapter 3]{stanley}.

\noindent {\bf Q.} Find obstructions to the existence of an $S_n$-equivariant meet-semilattice structure on $RP(X)$ that satisfies Equation \ref{eq:semil}.

\end{itemize}

\appendix

\section{Proofs}
\label{sec:proofs}

In this appendix, we describe the computer programs used to prove Theorem \ref{thm:ES} and Lemma \ref{lem:C2}.
Source code is freely available online, see \citep{phylogenetic-consensus}.
We will make use of standard terminology from group theory (see for example \citep{rotman2012introduction}).
When discussing associative stability,
we will also make use of standard terminology from the theory of partially ordered sets (see for example \citep{davey2002introduction}).

Given a finite set $X$, let us denote by $\sym(X)$ the symmetric group on $X$.
Then $\sym(X)$ acts naturally on the set $RP(X)$, by permuting the labels of the leaves.
Given a permutation $\sigma\in \sym(X)$ and a tree $T\in RP(X)$, we write $\sigma(T)$ for the tree obtained from $T$ by permuting the labels of the leaves according to $\sigma$.
In particular, this induces an action of the symmetric group on the set $RBP(X)$ of rooted binary phylogenetic trees.
The neutrality axiom can be restated by saying that the $k$-consensus functions $\varphi_X^k \colon RP(X)^k \to RP(X)$ should be equivariant with respect to the action of $\sym(X)$.

\subsection{Extension stability on binary trees}
As was discussed just after the Theorem's statement, in order to prove Theorem \ref{thm:ES} it is enough to check that there is no extension stable 2-consensus method among profiles of binary trees on a set $X$ of 5 taxa.
Since there is only a finite number of consensus methods on $X$, it is possible (at least in principle) to check every such consensus method by means of a computer program.
The number of 2-consensus functions $\varphi_X^2 \colon RBP(X)^2 \to RP(X)$ is, however, intractably large already for a set $X$ of cardinality 5.
In order to obtain an answer in a reasonable amount of time, we formulate our problem in the context of integer linear programming (see \citep{papadimitriou1998combinatorial}).

Consider the following set $\T$ of triples of phylogenetic trees:
\[ \T = \{ (T,T_1,T_2) \in RP(Y) \times RBP(Y) \times RBP(Y) \mid Y \subseteq X \}. \]
For every triple $(T,T_1,T_2) \in \T$, we introduce a boolean variable $m_{T,T_1,T_2} \in \{0,1\}$.
Denote by $\M$ the set of all these boolean variables.
A 2-consensus method $\varphi$ on $X$ (restricted to profiles of binary trees) corresponds to the following assignment of the variables in $\M$:
\[
m_{T,T_1,T_2} =
\begin{cases}
1 & \text{if } \varphi(T_1,T_2) = T; \\
0 & \text{otherwise}.
\end{cases}
\]
Conversely, an assignment of the variables in $\M$ yields a 2-consensus method on $X$ (restricted to binary trees), provided that the following linear relations are satisfied:
\[
\sum_{T \in RP(Y)} m_{T,T_1,T_2} = 1 \quad \text{for all $T_1,T_2 \in RBP(Y)$, for all $Y \subseteq X$.}
\]

Our aim is now to enforce all the requirements for our consensus method by means of linear equalities or inequalities involving the variables in $\M$.

\begin{enumerate}
	\item[(1)] Unanimity is equivalent to the following set of direct assignments:
	\[ m_{T,T,T} = 1 \quad \text{for all $T \in RBP(Y)$, for all $Y \subseteq X$}. \]
	
	\item[(2)] Anonymity is enforced as follows:
	\[ m_{T,T_1,T_2} = m_{T,T_2,T_1} \quad \text{for all $(T,T_1,T_2) \in \T$.} \]
	
	\item[(3)] Neutrality is given by:
	\[ m_{T,T_1,T_2} = m_{\sigma(T), \, \sigma(T_1), \, \sigma(T_2)} \quad \text{for all $(T,T_1,T_2) \in \T$, for all $\sigma \in \sym(X)$.} \]
	
	\item[(4)] Extension stability is slightly more complicated to encode.
	Consider any triple $(T,T_1,T_2) \in \T$, and let $Y$ be the set of leaves of $T$.
	For every subset $Z \subsetneq Y$, and for every tree $T' \in RP(Z)$ such that $T' \not\preceq T\vert_Z$, we require that
	\[ m_{T,T_1,T_2} + m_{T', \, T_1\vert_Z, \, T_2\vert_Z} \leq 1. \]
	The reason is that, if $m_{T,T_1,T_2} = m_{T', \, T_1\vert_Z, \, T_2\vert_Z} = 1$, then
	\[ \varphi(T_1\vert_Z, T_2\vert_Z) = T' \not\preceq T\vert_Z = \varphi(T_1,T_2)\vert_Z \]
	which violates extension stability.
	Conversely, a violation of extension stability translates into having $m_{T,T_1,T_2} = m_{T', \, T_1\vert_Z, \, T_2\vert_Z} = 1$ for some trees $T, T', T_1, T_2$ with $T' \not\preceq T\vert_Z$.
\end{enumerate}

\begin{remark}
	\label{rmk:optimizations}
	In a practical implementation, the equalities given in (1), (2), and (3) can be used to greatly reduce the number of variables involved in the model.
	Indeed, instead of using one boolean variable $m_{T, T_1, T_2}$ for every triple $(T,T_1,T_2) \in \T$, we use one variable for every orbit $[T, T_1, T_2] \in \T \,/\, (\sym(X) \times \Z_2)$. Here the action of $\sym(X) \times \Z_2$ on $\T$ is as follows: a permutation $\sigma \in \sym(X)$ maps $(T, T_1, T_2)$ to $(\sigma(T), \sigma(T_1), \sigma(T_2))$; the generator of $\Z_2$ maps $(T, T_1, T_2)$ to $(T, T_2, T_1)$.
	From a practical point of view, this means that we choose a representative for each orbit, and rewrite all the linear constraints in terms of variables $m_{T,T_1,T_2}$ where $(T,T_1,T_2)$ is the representative of its orbit.
	In addition, thanks to (1), we can remove all the variables of the form $m_{T,T',T'}$ for $T \neq T'$, because their value must be $0$.
	The optimizations described here are essential, in order to make the number of variables tractable.
\end{remark}

After optimizations, we obtain a model consisting of 11,688 boolean variables.
We use the solver Gurobi \citep{gurobi} to check that there exists no assignment of the variables that satisfies all the previous constraints.
Our program runs in approximately 4 minutes on a laptop with an Intel Core i7 processor ($8 \times 2.80$ GHz) and 16 GB of RAM.
This running time includes both the computation of the model and the proof of infeasibility.

\subsection{Associative stability}
\label{sec:ASproof}

Again as discussed after the Theorem's statement, in order to prove Theorem \ref{thm:AS} it is sufficient to show that there exists no regular associative 2-consensus function on a set $X$ of 5 taxa which is Pareto on rooted triples.
Just as in the case of extension stability, we formulate our problem in the context of integer linear programming.

The set of triples that we need to consider is simply $\T = RP(X)^3$ in this case.
For every triple $(T,T_1,T_2) \in \T$, we introduce a boolean variable $m_{T,T_1,T_2} \in \{0,1\}$ with the same meaning as in the previous section.
Again, denote by $\M$ the set of all these boolean variables.

As before, we need to express the fact that to every input $(T_1,T_2)$ corresponds a unique output $T$. We ensure this by requiring the following linear relations to be satisfied:
\[ \sum_{T \in RP(X)} m_{T,T_1,T_2} = 1 \quad \text{for all $T_1,T_2 \in RP(X)$}. \]
Assignments of the variables in $\M$ satisfying the previous relations are in one-to-one correspondence with 2-consensus functions $\varphi_X^2 \colon RP(X)^2 \to RP(X)$. 

Unanimity, anonymity, and associative stability of $\varphi_X^2$, are equivalent to $RP(X)$ being endowed with a partial order relation $\leq$ (not to be confused with the previously defined $\preceq$), such that every pair of trees $T_1,T_2 \in RP(X)$ has a unique greatest lower bound, given precisely by the tree $\varphi_X^2(T_1,T_2)$, see \cite[Chapter 1, p.~7]{Gra}.
Notice that, in particular, $T_1 \leq T_2$ if and only if $\varphi_X^2(T_1,T_2) = T_1$. The validity of the latter expression is represented by the value of the variable $m_{T_1,T_1,T_2}$; this leads us to introduce new variables 
\[ p_{T_1,T_2} := m_{T_1,T_1,T_2} \quad \text{for $T_1,T_2 \in RP(X)$}, \]
with the following meaning:
\[
p_{T_1,T_2} =
\begin{cases}
1 & \text{if $T_1 \leq T_2$}; \\
0 & \text{otherwise}.
\end{cases}
\]
Notice that the variables $p_{T_1,T_2}$ are simply aliases for some variables in $\M$.

We are now ready to translate all requirements for our consensus function into linear constraints.

\begin{enumerate}
	\item[(1)] Reflexive property of the partial order $\leq$:
	\[ p_{T,T} = 1 \quad \text{for all $T \in RP(X)$.} \]
	Notice that this set of assignments is equivalent to unanimity.
	
	\item[(2)] Antisymmetric property of the partial order $\leq$:
	\[ p_{T_1,T_2} + p_{T_2,T_1} \leq 1 \quad \text{for all $T_1,T_2 \in RP(X)$ with $T_1 \neq T_2$.} \]
	
	\item[(3)] Transitive property of the partial order $\leq$:
	\[ p_{T_1,T_3} \geq p_{T_1,T_2} + p_{T_2,T_3} - 1 \quad \text{for all $T_1,T_2,T_3 \in RP(X)$.} \]
	
	\item[(4)] The tree $\varphi_X^2(T_1,T_2)$ must be a lower bound of $T_1$ and $T_2$:
	\[ p_{T,T_1} \geq m_{T,T_1,T_2} \; \text{and} \; p_{T,T_2} \geq m_{T,T_1,T_2} \quad \text{for all $T,T_1,T_2 \in RP(X)$.} \]
	
	\item[(5)] The tree $\varphi_X^2(T_1,T_2)$ must be greater than every lower bound of $T_1$ and $T_2$:
	\[ m_{T,T_1,T_2} + p_{T',T_1} + p_{T',T_2} \leq p_{T',T} + 2  \quad \text{for all $T,T',T_1,T_2 \in RP(X)$.} \]
	Indeed, the only way to violate this constraint is to set $m_{T,T_1,T_2} = 1$ (that is, $\varphi_X^2(T_1,T_2) = T$), $p_{T',T_1} = p_{T',T_2} = 1$ (that is, $T'$ is a lower bound of $T_1$ and $T_2$), and $p_{T',T} = 0$ (that is, $T' \not\leq T$).
	
	\item[(6)] Neutrality:
	\[ m_{T,T_1,T_2} = m_{\sigma(T), \, \sigma(T_1), \, \sigma(T_2)} \quad \text{for all $T,T_1,T_2 \in RP(X)$.} \]
	
	\item[(7)] Pareto property on rooted triples:
	\[ m_{T,T_1,T_2} = 0 \quad \text{if $T_1\vert_Y = T_2\vert_Y \not\preceq T\vert_Y$ for some $Y\subseteq X$ with $|Y|=3$.} \]
\end{enumerate}

\begin{remark}
	As for extension stability, in our actual implementation we significantly reduce the number of variables.
	First, we only use one variable $m_{T, T_1, T_2}$ for every orbit with respect to the action of $\sym(X) \times \Z_2$, as described in Remark \ref{rmk:optimizations}.
	We also discard the variables of the form $m_{T, T', T'}$ with $T \neq T'$, because their value must be $0$.
	A further observation is that $T \not< \sigma(T)$ for every $T \in RP(X)$ and $\sigma \in \sym(X)$: in fact, if $T < \sigma(T)$, then using neutrality we obtain the contradiction
	\[ T < \sigma(T) < \sigma(\sigma(T)) < \dotsb < \sigma^k(T) = T, \]
	where $k$ is the order of $\sigma$.
	Therefore we can discard all the variables $m_{T, T_1, T_2}$ where $T$ is in the same $\sym(X)$-orbit of $T_1$ (respectively $T_2$) and $T \neq T_1$ (respectively $T \neq T_2$), because their value must be $0$ (recall that, if $m_{T, T_1, T_2} = 1$, then $T \leq T_1$ and $T \leq T_2$).
	Finally, we can discard in advance all the variables appearing in (7), because their value is $0$.
\end{remark}

With the optimizations described above, we end up with a model having 15,878 boolean variables.
As for extension stability, we use Gurobi \citep{gurobi} to check that there is no assignment of the variables in $\M$ for which all the previous constraints are satisfied.
Our program runs in approximately 8 minutes on a laptop with an Intel Core i7 processor ($8 \times 2.80$ GHz) and 16 GB of RAM.

\section{Consensus methods on small sets of taxa}
\label{sec:small-sets}

For every set $X$ of at most 4 taxa there is a unique regular associative consensus function on $X$ that is Pareto on rooted triples, namely Adams consensus. %
This can be checked using a variant of the program described in the Appendix.
The corresponding partial order relation $\leq$ discussed earlier
is represented in Figures \ref{fig:adams3} and \ref{fig:adams4}, for $X=\{1,2,3\}$ and $X=\{1,2,3,4\}$ respectively.

Notice that Adams consensus on a set $X$ of (at most) 4 taxa also satisfies extension stability, not only on binary trees, see \citep{Steel_1}.
However, it is not the only consensus method on $X$ which satisfies extension stability on binary trees.

\newcommand{\treeA}{
	\begin{tikzpicture}[x=1em,y=1em]
	\node (R) at (0,2) {};
	\node (1) at (-1.5,0) {\strut 1};
	\node (2) at (0,0) {\strut 2};
	\node (3) at (1.5,0) {\strut 3};
	\draw (1.north) -- (R.center);
	\draw (2.north) -- (R.center);
	\draw (3.north) -- (R.center);
	\end{tikzpicture}
}

\newcommand{\treeB}[3]{
	\begin{tikzpicture}[x=1em,y=1em]
	\node (R) at (0,3) {};
	\node (A) at (-0.75,2) {};
	\node (1) at (-1.5,0) {\strut #1};
	\node (2) at (0,0) {\strut #2};
	\node (3) at (1.5,0) {\strut #3};
	\draw (1.north) -- (A.center);
	\draw (2.north) -- (A.center);
	\draw (3.north) -- (R.center);
	\draw (A.center) -- (R.center);
	\end{tikzpicture}
}

\begin{figure}[p]
	{\small
		\begin{center}
			\begin{tikzpicture}
			\node (A) at (0,0) {\treeA};
			\node (B1) at (-2,2.2) {\treeB{1}{2}{3}};
			\node (B2) at (0,2.2) {\treeB{1}{3}{2}};
			\node (B3) at (2,2.2) {\treeB{2}{3}{1}};
			\draw (0,0.7) -- (B1.south);
			\draw (0,0.7) -- (B2.south);
			\draw (0,0.7) -- (B3.south);
			\end{tikzpicture}
		\end{center}
	}
	\caption{Partial order $\leq$ associated with Adams consensus, for $X=\{1,2,3\}$. The consensus tree of $T_1$ and $T_2$ is the highest common descendant of $T_1$ and $T_2$.}
	\label{fig:adams3}
\end{figure}
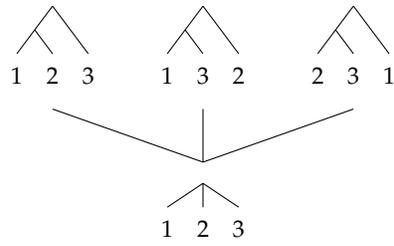

\newcommand{\leaves}[4]{
	\node (1) at (-1.8,0) {\strut #1};
	\node (2) at (-0.6,0) {\strut #2};
	\node (3) at (0.6,0) {\strut #3};
	\node (4) at (1.8,0) {\strut #4};
}

\newcommand{\treeC}{
	\begin{tikzpicture}[x=1em,y=1em]
	\node (R) at (0,2.1) {};
	\leaves{1}{2}{3}{4}
	\draw (1.north) -- (R.center);
	\draw (2.north) -- (R.center);
	\draw (3.north) -- (R.center);
	\draw (4.north) -- (R.center);
	\end{tikzpicture}
}

\newcommand{\treeD}[4]{
	\begin{tikzpicture}[x=1em,y=1em]
	\node (R) at (0,3) {};
	\node (A) at (-1.2,1.7) {};
	\leaves{#1}{#2}{#3}{#4}
	\draw (1.north) -- (A.center);
	\draw (2.north) -- (A.center);
	\draw (3.north) -- (R.center);
	\draw (4.north) -- (R.center);
	\draw (A.center) -- (R.center);
	\end{tikzpicture}
}

\newcommand{\treeE}[4]{
	\begin{tikzpicture}[x=1em,y=1em]
	\node (R) at (0,3) {};
	\node (A) at (-1.2,1.7) {};
	\node (B) at (1.2,1.7) {};
	\leaves{#1}{#2}{#3}{#4}
	\draw (1.north) -- (A.center);
	\draw (2.north) -- (A.center);
	\draw (3.north) -- (B.center);
	\draw (4.north) -- (B.center);
	\draw (A.center) -- (R.center);
	\draw (B.center) -- (R.center);
	\end{tikzpicture}
}

\newcommand{\treeF}[4]{
	\begin{tikzpicture}[x=1em,y=1em]
	\node (R) at (0.1,2.7) {};
	\node (A) at (-0.6,2.1) {};
	\leaves{#1}{#2}{#3}{#4}
	\draw (1.north) -- (A.center);
	\draw (2.north) -- (A.center);
	\draw (3.north) -- (A.center);
	\draw (4.north) -- (R.center);
	\draw (A.center) -- (R.center);
	\end{tikzpicture}
}

\newcommand{\treeG}[4]{
	\begin{tikzpicture}[x=1em,y=1em]
	\node (R) at (0,3.1) {};
	\node (A) at (-0.6,2.4) {};
	\node (B) at (-1.2,1.7) {};
	\leaves{#1}{#2}{#3}{#4}
	\draw (1.north) -- (B.center);
	\draw (2.north) -- (B.center);
	\draw (3.north) -- (A.center);
	\draw (4.north) -- (R.center);
	\draw (A.center) -- (R.center);
	\draw (B.center) -- (A.center);
	\end{tikzpicture}
}

\begin{figure}[p]
	{\small
		\begin{center}
			\bigskip
			\begin{tikzpicture}
			\node (A) at (0,0) {\treeC};
			\node (B1) at (-4.25,2.6) {\treeD{1}{2}{3}{4}};
			\node (B2) at (-2.55,2.6) {\treeD{1}{3}{2}{4}};
			\node (B3) at (-0.85,2.6) {\treeD{1}{4}{2}{3}};
			\node (B4) at (0.85,2.6) {\treeD{2}{3}{1}{4}};
			\node (B5) at (2.55,2.6) {\treeD{2}{4}{1}{3}};
			\node (B6) at (4.25,2.6) {\treeD{3}{4}{1}{2}};
			\node (C1) at (-5.1,5.2) {\treeF{1}{2}{3}{4}};
			\node (C2) at (-3.4,5.2) {\treeF{1}{2}{4}{3}};
			\node (C3) at (3.4,5.2) {\treeF{1}{3}{4}{2}};
			\node (C4) at (5.1,5.2) {\treeF{2}{3}{4}{1}};
			\node (D1) at (1.7,5.2) {\treeE{1}{2}{3}{4}};
			\node (D2) at (-1.7,5.2) {\treeE{1}{3}{2}{4}};
			\node (D3) at (0,5.2) {\treeE{1}{4}{2}{3}};
			\node (E1) at (-3,7.8) {\treeG{$x$}{$y$}{$z$}{4}};
			\node (E2) at (-1,7.8) {\treeG{$x$}{$y$}{$z$}{3}};
			\node (E3) at (1,7.8) {\treeG{$x$}{$y$}{$z$}{2}};
			\node (E4) at (3,7.8) {\treeG{$x$}{$y$}{$z$}{1}};
			\draw (0,0.6) -- (B1.south);
			\draw (0,0.6) -- (B2.south);
			\draw (0,0.6) -- (B3.south);
			\draw (0,0.6) -- (B4.south);
			\draw (0,0.6) -- (B5.south);
			\draw (0,0.6) -- (B6.south);
			\draw (B1.north) -- (C1.south);
			\draw (B1.north) -- (C2.south);
			\draw (B2.north) -- (C1.south);
			\draw (B2.north) -- (C3.south);
			\draw (B3.north) -- (C2.south);
			\draw (B3.north) -- (C3.south);
			\draw (B4.north) -- (C1.south);
			\draw (B4.north) -- (C4.south);
			\draw (B5.north) -- (C2.south);
			\draw (B5.north) -- (C4.south);
			\draw (B6.north) -- (C3.south);
			\draw (B6.north) -- (C4.south);
			\draw (B1.north) -- (D1.south);
			\draw (B2.north) -- (D2.south);
			\draw (B3.north) -- (D3.south);
			\draw (B4.north) -- (D3.south);
			\draw (B5.north) -- (D2.south);
			\draw (B6.north) -- (D1.south);
			\draw (-4.9,5.9) -- (E1.south);
			\draw (-3.2,5.9) -- (E2.south);
			\draw (C3.north) -- (E3.south);
			\draw (C4.north) -- (E4.south);
			\end{tikzpicture}
		\end{center}
	}
	\caption{Partial order $\leq$ associated with Adams consensus, for $X=\{1,2,3,4\}$.
		Each of the four elements on the top represents three different trees, obtained by choosing the values of $x,y,z$ in all possible ways.}
	\label{fig:adams4}
\end{figure}
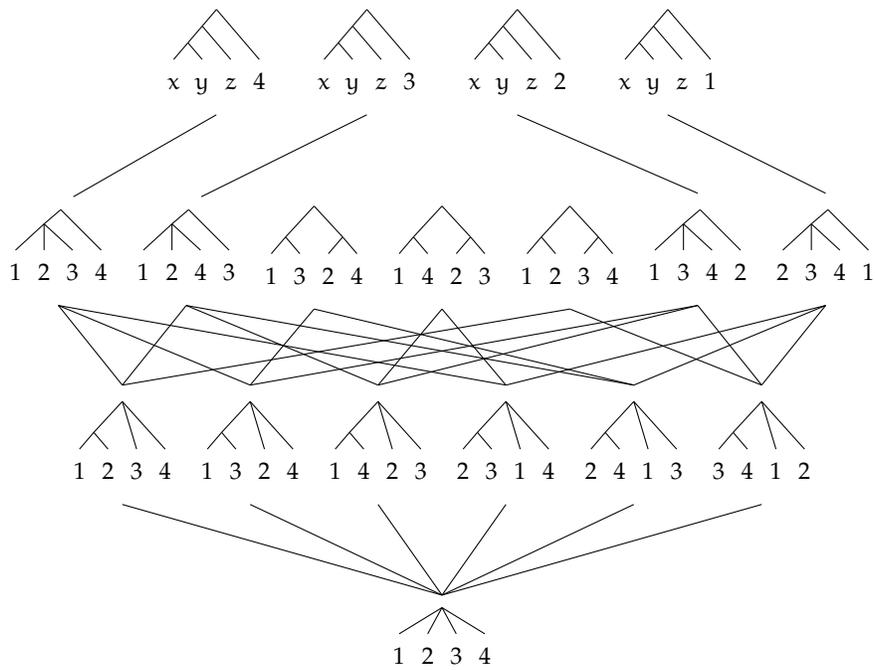

\section*{Acknowledgements} Emanuele Delucchi and Linard Hoessly were supported by the Swiss National Science Foundation Professorship grant PP00P2\_179110/1.

\bibliography{references} 

\begin{thebibliography}{10}

\bibitem{Adams}
Edward~N. Adams.
\newblock N-trees as nestings: Complexity, similarity, and consensus.
\newblock {\em J. Classif.}, 3(2):299--317, Sep 1986.

\bibitem{Aho}
Alfred~V. Aho, Yehoshua Sagiv, Thomas~G. Szymanski, and Jeffrey~D. Ullman.
\newblock Inferring a tree from lowest common ancestors with an application to
  the optimization of relational expressions.
\newblock {\em SIAM J. Comput.}, 10(3):405--421, 1981.

\bibitem{Arrow}
Kenneth~J. Arrow.
\newblock A difficulty in the concept of social welfare.
\newblock {\em J. Polit. Econ.}, 58(4):328--346, 1950.

\bibitem{Barth}
Jean-Pierre Barth\'el\'emy, Fred~R. McMorris, and Robert~C. Powers.
\newblock Stability conditions for consensus functions defined on {$n$}-trees.
\newblock {\em Math. Comput. Modelling}, 22(1):79--87, 1995.

\bibitem{Bryant_classif}
David Bryant.
\newblock A classification of consensus methods for phylogenetics.
\newblock In {\em Bioconsensus}, volume~61 of {\em DIMACS Series in Discrete
  Mathematics and Theoretical Computer Science}, pages 163--183. Amer. Math.
  Soc., Providence, RI, 2003.

\bibitem{Steel_1}
David Bryant, Andrew Francis, and Mike Steel.
\newblock Can we ``future-proof'' consensus trees?
\newblock {\em Syst. Biol.}, 66(4):611--619, 2017.

\bibitem{davey2002introduction}
Brian~A. Davey and Hilary~A. Priestley.
\newblock {\em Introduction to lattices and order}.
\newblock Cambridge university press, 2002.

\bibitem{McMorris_book}
William H.~E. Day and Fred~R. McMorris.
\newblock {\em Axiomatic Consensus Theory in Group Choice and Biomathematics}.
\newblock Frontiers in Applied Mathematics. Society for Industrial and Applied
  Mathematics, 2003.

\bibitem{degnan}
James~H. Degnan, Michael DeGiorgio, David Bryant, and Noah~A. Rosenberg.
\newblock Properties of consensus methods for inferring species trees from gene
  trees.
\newblock {\em Syst. Biol.}, 58(1):35--54, 2009.

\bibitem{phylogenetic-consensus}
Emanuele Delucchi, Linard Hoessly, and Giovanni Paolini.
\newblock Phylogenetic consensus.
\newblock \url{https://github.com/giove91/phylogenetic-consensus}, 2018.
\newblock GitHub repository.

\bibitem{dong2010}
Jianrong Dong, David Fern{\'a}ndez-Baca, and Fred~R McMorris.
\newblock Constructing majority-rule supertrees.
\newblock {\em Algorithms for Molecular Biology}, 5(1):2, 2010.

\bibitem{Gra}
George Gr\"{a}tzer.
\newblock {\em General lattice theory}.
\newblock Birkh\"{a}user Verlag, Basel-Stuttgart, 1978.
\newblock Lehrb\"{u}cher und Monographien aus dem Gebiete der Exakten
  Wissenschaften, Mathematische Reihe, Band 52.

\bibitem{gurobi}
{Gurobi Optimization LLC}.
\newblock Gurobi optimizer reference manual, 2018.

\bibitem{Husonetal}
Daniel~H. Huson, Regula Rupp, and Celine Scornavacca.
\newblock {\em Phylogenetic networks}.
\newblock Cambridge University Press, 2010.

\bibitem{papadimitriou1998combinatorial}
Christos~H. Papadimitriou and Kenneth Steiglitz.
\newblock {\em Combinatorial optimization: algorithms and complexity}.
\newblock Courier Corporation, 1998.

\bibitem{rotman2012introduction}
Joseph~J. Rotman.
\newblock {\em An introduction to the theory of groups}, volume 148.
\newblock Springer, 2012.

\bibitem{stanley}
Richard~P. Stanley.
\newblock {\em Enumerative combinatorics. {V}olume 1}, volume~49 of {\em
  Cambridge Studies in Advanced Mathematics}.
\newblock Cambridge University Press, Cambridge, second edition, 2012.

\bibitem{steel_book16}
Mike Steel.
\newblock {\em Phylogeny: Discrete and Random Processes in Evolution}.
\newblock CBMS-NSF Regional Conference Series on Mathematics. Society for
  Industrial and Applied Mathematics, 2016.

\bibitem{Steel_2}
Mike Steel and Joel~D. Velasco.
\newblock Axiomatic opportunities and obstacles for inferring a species tree
  from gene trees.
\newblock {\em Syst. Biol.}, 63(5):772--778, 2014.

\end{thebibliography}
\bibliographystyle{plain}
 
\end{document}